\documentclass[plain]{aiaa-pretty}
\usepackage{amsthm}
\usepackage{amsmath}
\usepackage{amsfonts}
\usepackage{amssymb}
\usepackage{dsfont}
\usepackage{subfigure}
\usepackage{psfrag}
\usepackage{balance}
\usepackage{float}
\usepackage[titletoc,title]{appendix}
\usepackage{booktabs, makecell}
\usepackage{threeparttable}
\usepackage{multirow}
\usepackage{bm}
\usepackage[noend]{algpseudocode}
\usepackage[ruled]{algorithm}
\let\doi\relax
\usepackage{doi}

\newcommand\boldhat{\bm\hat}

\newtheorem{proposition}{Proposition}
\newtheorem{remark}{Remark}

\graphicspath{{Figs/}}

\author{
Mirko Leomanni\thanks{Research Associate, Department of Information Engineering and Mathematics, \texttt{$\!$leomanni@diism.unisi.it}},
Renato Quartullo \thanks{PhD Student, Department of Information Engineering and Mathematics, \texttt{quartullo@diism.unisi.it}},
Gianni Bianchini\thanks{Professor, Department of Information Engineering and Mathematics, \texttt{giannibi@diism.unisi.it}},
\\
Andrea Garulli \thanks{Professor, Department of Information Engineering and Mathematics, \texttt{garulli@diism.unisi.it}},
Antonio Giannitrapani\thanks{Professor, Department of Information Engineering and Mathematics, \texttt{giannitrapani@diism.unisi.it}}\\[1mm]
\textit{University of Siena, Siena, 53100, Italy}
}

\title{Variable-Horizon Guidance for Autonomous\\ Rendezvous and Docking to a Tumbling Target}

\abstract{In this paper, the trajectory planning problem for autonomous rendezvous and docking between a controlled spacecraft and a tumbling target is addressed. The use of a variable planning horizon is proposed in order to construct an appropriate maneuver plan, within an optimization-based framework. The involved optimization problem is nonconvex and features nonlinear constraints. The main contribution is to show that such problem can be tackled effectively by solving a finite number of linear programs. To this aim, a specifically conceived horizon search algorithm is employed in combination with a polytopic constraint approximation technique. The resulting guidance scheme provides the ability to identify favourable docking configurations, by exploiting the time-varying nature of the optimization problem endpoint. Simulation results involving the capture of the nonoperational EnviSat spacecraft indicate that the method is able to generate optimal trajectories at a fraction of the computational cost incurred by a state-of-the-art nonlinear solver.}

\begin{document}
\maketitle

\section{Introduction}

Spacecraft rendezvous and docking (RVD) technologies were tested for the first time in the 1960s within the Gemini and Soyuz programs and later brought to operational status with the advent of manned space stations \cite{woffinden07}. In these programs, RVD was accomplished through manual or semi-automated procedures involving a tight cooperation among the vehicles, heavy instrumentation, and man or ground-in-the-loop interaction to ensure successful maneuvering. In recent years, a demand for new RVD technologies has emerged in the context of small multi-purpose servicing vehicles, which will be capable of autonomously performing a number of complex tasks such as re-fueling, in-orbit repair/assembly and orbital debris capture. Several technology demonstration missions have been carried out for servicing a three-axis stabilized spacecraft, including JAXA’s ETS-VII \cite{kawano01}, NASA’s DART \cite{rumford2003demonstration}, AFRL's XSS-11 \cite{mitchell2006gnc}, and DARPA’s Orbital Express \cite{weismuller2006gn}. However, autonomous RVD to an uncontrolled and possibly tumbling (i.e., rotating) target has yet to be fully demonstrated in orbit, and servicing or capturing an uncooperative target still involves a number of open problems. In particular, there is a need for RVD techniques accounting for rotational motion of the target, which optimize meaningful performance indexes and are easy to implement onboard the spacecraft.

Achieving autonomous RVD involves many complementary operations: inspection, pose estimation, maneuver planning, attitude synchronization, and relative motion control. From a guidance and control perspective, the main challenge to be faced when the target is uncooperative is that the docking point position may vary over time. This leads to the formulation of trajectory planning problems in which both the endpoint and the constraints are time-varying. A wide variety of optimization-based techniques have been proposed in the literature to tackle such type of problems, which in general entail nonlinear optimization methods. In \cite{boyarko2011}, minimum-time and minimum-energy rendezvous trajectories are obtained via direct collocation and validated against the first-order optimality conditions provided by the Pontryagin minimum principle. The approach is refined in \cite{ventura2017} to limit the computational burden. The primary advantage of these nonlinear programming (NLP) methods is that they can account for nonlinearities in the system model, nonconvex constraints, and a free final time. However, NLP is affected by a number of well-known drawbacks, including the lack of convergence guarantees and the requirement of complex solvers. In order to mitigate these issues, researchers in the field have focused on sequential convex programming (SCP) \cite{ping13, liu13, mao2019successive, bonalli19gusto}. Within this approach, all nonconvex elements of the trajectory optimization problem are linearized and the resulting convex problem is solved in a trust region where the linearization is accurate. The process is repeated iteratively until a stopping criterion is met. For certain classes of problems, the asymptotic convergence to a local optimum has been proved \cite{liu14}. However, for general nonlinear problems, the solution sequence may converge to an infeasible trajectory \cite{malyuta2021convex}. This is one of the salient limitations of SCP, and an intensive research is ongoing to overcome this obstacle (see, e.g., \cite{Foust20,ping21}). At present, SCP provides an effective and flexible way to perform rapid trajectory optimization trade studies, see, e.g., \cite{malyuta2020}.

In order to facilitate online optimization, a great deal of research has been directed towards problem formulations which are inherently convex. Convex formulations are usually obtained by linearizing the spacecraft relative motion dynamics, exploiting suitable convex approximations of the RVD path constraints, and adopting a fixed planning horizon \cite{weiss15,zagaris18}. Most of the studies in this area focus on RVD to a cooperative target, assuming that the docking point is static, see, e.g., \cite{hartley13,leomanni20,mammarella20}. Some important contributions have addressed the uncooperative RVD problem. In \cite{breger08safe}, a convex description of the path constraints is introduced and shown to provide much faster solutions compared to a mixed-integer linear programming (MILP) formulation of such constraints (see, e.g., \cite{richards02}). In \cite{dicariano2012}, a model predictive control (MPC) strategy is applied to a planar RVD problem with a tumbling target. This is extended to the three-dimensional case in \cite{li17,dong20}. Although such techniques have nowadays proven to be suitable for implementation onboard a spacecraft (see, e.g., \cite{mammarella20}), their application to RVD missions still faces remarkable challenges. For instance, the use of a fixed planning horizon may not be consistent with the mission requirements, as it prevents from taking into account the maneuver time in the cost function. This has been first pointed out in \cite{richards2006robust,hartley2012}, where a variable-horizon formulation is proposed in order to improve the regulation performance, for RVD maneuvers involving a three-axis stabilized target. However, to the best of our knowledge, variable-horizon approaches tailored to the case of tumbling targets have not been developed to date. This is a serious limitation because the planning horizon dictates the optimization problem endpoint (being the docking point time-varying) and thus the whole maneuver geometry. Therefore, an improper choice of this parameter can lead to severe performance degradation or even infeasibility. In light of these considerations, and taking into account that the characteristics of the target rotational motion are generally not known beforehand, it is reasonable to expect that the horizon length will have to be tuned on orbit. This makes it necessary to adopt a variable-horizon strategy, to be run in real time onboard the spacecraft.

Variable-horizon optimal control problems can be addressed either in a continuous-time or in a discrete-time setting. In the former, a free-final-time problem is converted into a fixed-final-time one by normalizing the time variable. The resulting optimization problem is nonlinear, even for linear dynamical systems. The latter approach amounts to solving a sequence of fixed-horizon problems, in which linearity of the dynamics is preserved. This simplifies the convergence analysis. However, treating the horizon length as an additional decision variable leads to mixed-integer optimization problems that are difficult to solve (see, e.g., \cite{richards2006robust}).

The contribution of this paper is to provide an effective discrete-time solution to the variable-horizon guidance problem, for RVD to a tumbling target. The only source of nonconvexity in the proposed formulation is due to the variable horizon, which is treated as an assignable parameter and weighted in the cost function of the trajectory optimization problem. For any value of the horizon length, the optimization problem is cast a linear program (LP). This is achieved by suitably approximating the RVD constraints. In particular, nonconvex keep-out-zone constraints are approximated by a set of linear time-varying inequalities, using a variant of the so-called rotating hyperplane strategy (see, e.g., \cite{zagaris18}). Then, the solution to the variable-horizon problem is obtained by solving a finite number of LPs. The use of rotating hyperplanes is instrumental to mitigate the computational burden. In fact, for any given horizon length, the proposed approach requires to solve a single LP, while alternative (and less conservative) methods based on constraint linearization (see, e.g., \cite{liu14}) involve the solution of a sequence of convex programs. Another key element of the proposed solution strategy is the construction of a convenient initial guess for the horizon length, around which a local search is performed. The resulting optimization algorithm ensures convergence to a local optimum in a finite and typically small number of steps. A parametric study shows the advantages of this approach with respect to other solution techniques commonly employed for variable-horizon optimization.

The guidance scheme is demonstrated on two simulated maneuvers inspired by the capture of the nonoperational EnviSat spacecraft \cite{louet99}. Simulation results show that the method is able to generate safe RVD trajectories to the tumbling target at a fraction of the computational cost incurred by a state-of-the-art nonlinear solver. Moreover, the obtained results indicate that the employed constraint approximation scheme, although conservative, does not lead to a significant loss in terms of maneuver performance. These features make the proposed approach attractive for autonomous RVD applications, in which the solution to the guidance problem must be computed onboard the spacecraft.

The paper is organized as follows. The variable-horizon guidance problem is formulated in Section~\ref{pbset}, and the RVD constraint model is presented in Section~\ref{docking}. The proposed solution strategy is discussed in Section~\ref{solution}. In Sections~\ref{SAp}-\ref{Ervd} the performance of the method is evaluated numerically and the EnviSat RVD case studies are detailed, while conclusions are drawn in Section~\ref{conclusions}.

\paragraph{Notation:}
The adopted notation is fairly standard. The sets of real, nonnegative real, positive integer, and nonnegative integer numbers are denoted by $\mathbb{R}$, $\mathbb{R}^+$, $\mathbb{N}$, and $\mathbb{N}_0$, respectively. The time derivative of a vector $\mathbf{x}$ is denoted by $\dot{\mathbf{x}}$. The $p$-norm of $\mathbf{x}$ and the direction of $\mathbf{x}$ are indicated by $\| \mathbf{x} \|_p$ and  $\vec{\mathbf{x}}=\mathbf{x}/\| \mathbf{x} \|_2$, respectively. The symbol $\times$ indicates the cross-product operation. The pseudoinverse of a matrix $\mathbf{M}$ is denoted by $\mathbf{M}^\dag$. The set difference operation is denoted by $\setminus$ and the empty set is denoted by $\varnothing$.  The matrix describing a rotation about the axis $\mathbf{a}_x\in\mathbb{R}^3$ by an angle $\theta\in \mathbb{R}$ is denoted by $\mathbf{R}(\mathbf{a}_x,\theta)$.

\par

\section{Problem Formulation}\label{pbset}
The considered spacecraft maneuvering problem is that of RVD between an actively controlled servicer vehicle and a tumbling target. Herein, the focus is on guidance aspects, i.e., on the generation of safe rendezvous and docking trajectories to the target. We restrict our attention to the motion of the center of mass of the servicer (sCM) relative to a time-varying docking point modeling the rotational motion of the target. Hence, the servicer attitude motion is neglected. The rationale behind this approach is that, under reasonable assumptions (namely, the availability of omnidirectional thrust), the translational dynamics of the servicer can be decoupled from its own attitude dynamics \cite{fehse2003}, thus resulting in a simplified guidance algorithm design.

The reference coordinate frame employed in this work is the Radial-Transverse-Normal (RTN) frame centered at the target. The R-axis is aligned to the target radius vector, the N-axis points towards the target orbit normal, and the T-axis completes a right handed triad. According to standard design rules, the dynamics of the sCM with respect to the center of mass of the target (tCM) are expressed in terms of relative position and velocity, by using the normalized discrete-time Hill-Clohessy-Wiltshire equations \cite{clohessy60}
\begin{equation}\label{zetacirctd}
\begin{array}{lll}
\mathbf{x}(k+1)&=&\mathbf{A} \mathbf{x}(k)+ \mathbf{B} \mathbf{\mathbf{u}}(k)\\[2mm]
&=&e^{ \mathbf{A}_c  \tau_s\,} \mathbf{x}(k) +\left(\displaystyle\int_{\tau=0}^{\tau_s}e^{ \mathbf{A}_c \tau}\text{d}\tau\right)\!\mathbf{B}_c\, \mathbf{\mathbf{u}}(k),
\end{array}
\end{equation}
where $k\in \mathbb{N}_0$ is the discrete time index, $\mathbf{x}(k)\in\mathbb{R}^6$ is the system state, $\mathbf{\mathbf{u}}(k)\in\mathbb{R}^3$ is the control input, $\tau_s\in\mathbb{R}^+$ is the sampling interval and
\begin{equation}\label{zetacircmattd}
\mathbf{A}_c=
\left[
\begin{array}{c c c c c c}
0&0&0&1&0&0\\
0&0&0&0&1&0\\
0&0&0&0&0&1\\
3&0&0&0&2&0\\
0&0&0&-2&0&0\\
0&0&-1&0&0&0
\end{array}\right]
\quad\quad
\mathbf{B}_c=\left[
\begin{array}{lll}
0&0 &0\\
0& 0 & 0\\
0&  0&   0 \\
1& 0 &    0 \\
0& 1& 0  \\
0& 0 &1
\end{array}\right]
\end{equation}
The following properties pertain to this representation. The control input is defined as $\mathbf{u}(k)={\mathbf{a}(k)}/{ a_{max}}$, being $\mathbf{a}(k)\in\mathbb{R}^3$ the actual acceleration, expressed in RTN coordinates, and $a_{max}$ the maximum acceleration deliverable by the servicer along each axis of the RTN frame ($a_{max}$ is assumed constant). Notice that, in this setting,
\begin{equation}\label{inpcon}
\| \mathbf{u}(k) \|_\infty\leq 1.
\end{equation}
The state vector is defined as $\mathbf{x}(k)=[\mathbf{x}_p^T(k)\;\mathbf{x}_v^T(k)]^T$, where $\mathbf{x}_p(k)\in\mathbb{R}^3$ describes the RTN components of the relative position vector multiplied by $\eta^2/a_{max}$, and $\mathbf{x}_v(k)\in\mathbb{R}^3$ describes the RTN components of the relative velocity vector multiplied by $\eta/a_{max}$, being $\eta$ the target mean motion. In \eqref{zetacirctd}, a scaled time variable $\tau=\eta t$ is employed, where $t\in\mathbb{R}^+$ is the actual time. Then, $t=k \tau_s/\eta$ at the sampling instants. The choice of the linear time-invariant model \eqref{zetacirctd}-\eqref{zetacircmattd} is appropriate for circular orbits and is made for ease of exposition. However, it is worth stressing that the method presented hereafter applies in general to linear time-varying models, such as the ones describing the relative motion in elliptical orbits.

\begin{figure}[h]
\centering
\includegraphics[width=0.3\textwidth]{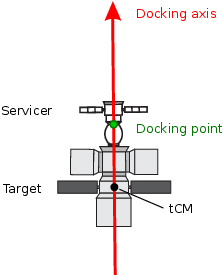}
\caption{Definition of the docking point.}
\label{dkp}
\end{figure}
The docking specifications are modeled by defining a docking axis which is rigidly attached to the target and specifying a suitable docking point along this axis, as illustrated in Fig.~\ref{dkp}. The docking point describes the desired position of the sCM upon docking. Taking into account the target rotational motion, it follows that the
position of the docking point evolves on a sphere of constant radius. The linear position $\mathbf{p}^d\in\mathbb{R}^3$ and velocity $\mathbf{v}^d\in\mathbb{R}^3$ of this point relative to the tCM satisfy the differential equation
\begin{equation}\label{eq2}
\begin{array}{lll}
\dot{\mathbf{p}}^d(t)&=&\mathbf{v}^d(t)\\[3mm]
\mathbf{v}^d(t)&=&\boldsymbol{\omega}(t) \times \mathbf{p}^d(t),
\end{array}
\end{equation}
where $\boldsymbol{\omega}(t)\in\mathbb{R}^3$ is the instantaneous angular velocity of the target body frame relative to the RTN frame. All vectors in \eqref{eq2} are expressed in the RTN frame.  The reference state trajectory $\mathbf{x}^d(k)\in\mathbb{R}^6$ for the guidance problem is obtained by sampling $\mathbf{p}^d(t)$, $\mathbf{v}^d(t)$ and applying the same normalization used to obtain $\mathbf{x}(k)$, resulting in
\begin{equation}\label{eq3}
\mathbf{x}^d(k)=
\left[
\begin{array}{l}
\mathbf{x}_p^d(k)\\[3mm] \mathbf{x}_v^d(k)
\end{array}
\right]
=
\left[
\begin{array}{l}
\dfrac{\eta^2}{a_{max}}\, \mathbf{p}^d\!\left(\dfrac{k \tau_s}{\eta}\right)\\[3mm] \dfrac{\eta}{a_{max}}\, \mathbf{v}^d\!\left(\dfrac{k \tau_s}{\eta}\right)
\end{array}
\right].
\end{equation}

The rendezvous and docking maneuver objective is stated as follows: steer the state $\mathbf{x}(k)$
of system \eqref{zetacirctd} from a given initial condition $\mathbf{x}_0$ at time $k_0$  to the reference trajectory $\mathbf{x}^d(k)$, while minimizing a trade-off between fuel consumption and maneuver time, and satisfying suitable state and input constraints. The proposed guidance scheme achieves this objective through the solution of the following variable-horizon discrete-time optimal control problem:
\begin{equation}\label{fpt2}
\begin{aligned}
\underset{N,\,\mathbf{u}_N}{\text{min}}~~ & J_N= N+ \gamma\|\mathbf{u}_N\|_1 \\[1mm]
\text{s.t.} \quad & {\mathbf{x}}(k_0)= \mathbf{x}_0\\[1mm]
& {\mathbf{x}}(k+1)=\mathbf{A} {\mathbf{x}}(k)+\mathbf{B}{\mathbf{u}}(k)  \\[1mm]
&  {\mathbf{x}}_p(k)\in \mathcal{X}(k,N) \quad k=k_0,\ldots, k_0+N-1\\[1mm]
& \|\mathbf{u}_N\|_\infty \leq 1 \\[1mm]
&  {\mathbf{x}}(k_0+N)=\mathbf{x}^d(k_0+N) \\[1mm]
&  N\in\mathbb{N}
\end{aligned}
\end{equation}
In problem \eqref{fpt2},
\begin{equation*}\label{usequence}
{\mathbf{u}}_N=\begin{bmatrix} {\mathbf{u}}(k_0+N-1)\\ \vdots  \\{\mathbf{u}}(k_0)\end{bmatrix}
\end{equation*}
is a control sequence of variable length $N$ to be optimized. The cost function $J_N$ is the same used in \cite{richards2006robust} and involves the sum of the number of steps $N$ required to steer the initial state $\mathbf{x}_0$ towards the docking state $\mathbf{x}^d(k_0+N)$ and the normalized fuel consumption $\|\mathbf{u}_N\|_1$, which is weighed by the scalar parameter $\gamma\geq0$. The constraint $\|\mathbf{u}_N\|_\infty \leq 1$ enforces a bound for each component of the control sequence according to \eqref{inpcon}. The state constraint sets $\mathcal{X}(k,N)$ are allowed to depend explicitly on both $k$ and $N$. They are assumed to be convex and polytopic for any fixed $k$ and $N$. The particular structure of $\mathcal{X}(k,N)$ is detailed in Section \ref{docking}. Problem \eqref{fpt2} is a nonstandard one in discrete-time optimal control theory, because the number of decision variables and constraints is dictated by the optimization variable $N$.

\section{Rendezvous and Docking Constraints}\label{docking}
In order to safely achieve the rendezvous and docking objective, collisions must be avoided. Moreover, the servicer position must be confined within a suitable visibility region during the final part of the maneuver. A common approach to address such requirements is to introduce a separate set of constraints for the rendezvous and the docking phases.
To this aim, we find it convenient to model the safety constraints enforced on the whole maneuver as $\mathbf{x}_p(k)\in \mathcal{X}(k,N)$, where
\begin{equation}\label{ssconstr3}
\mathcal{X}(k,N)=
\left\{
\begin{array}{ll}
\mathcal{R}(k,\lambda_N),&\quad k < k_0+  \lambda_N\\[2mm]
\mathcal{D}(k),&\quad k  \geq k_0+ \lambda_N
\end{array}
\right.
\end{equation}
In \eqref{ssconstr3}, $\mathcal{R}(k,\lambda_N)$ and $\mathcal{D}(k)$ denote the (time-varying) constraint sets for the rendezvous and the docking phase, respectively, and $\lambda_N$ is the number of sampling instants allocated to the rendezvous phase, within the planning horizon. The latter is treated as an horizon-dependent variable, by setting
\begin{equation}\label{lambdaN}
\lambda_N=N-N_d,
\end{equation}
where the parameter $N_{d}\in \mathbb{N}$ indicates a predefined number of time steps allocated to the docking phase, at the end of the planning horizon. Note that $N_{d}$ is fixed a-priori since docking operations must be usually completed within a fixed amount of time.

\begin{figure}[!t]
	\centering
		\psfrag{a}{\small$\alpha$}
		\psfrag{r}{\small$r$}
	\includegraphics[width=0.8\textwidth]{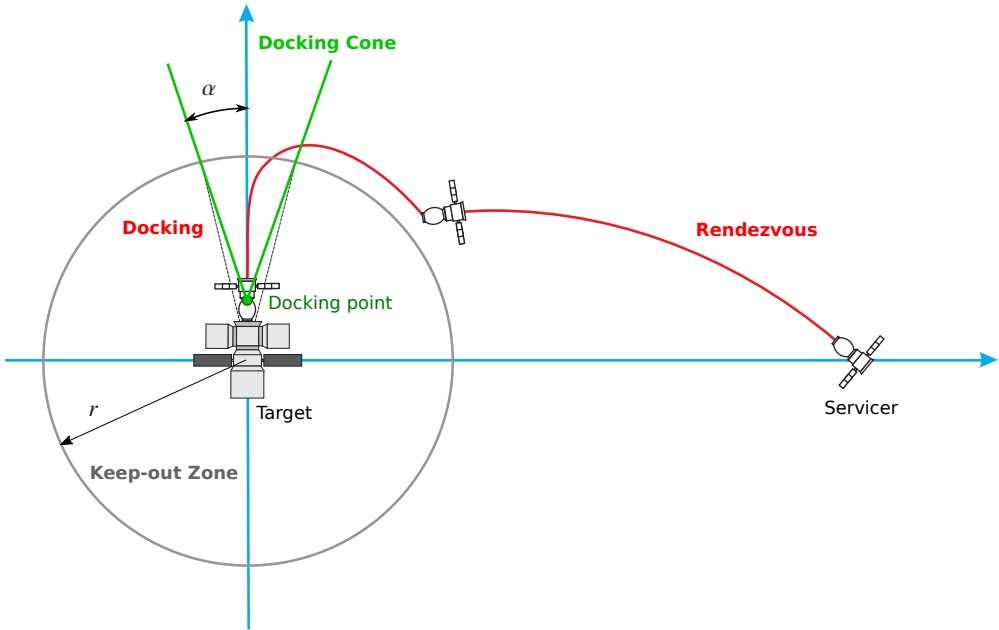}
	\caption{Illustration of a typical rendezvous and docking maneuver. For the sake of illustration, the docking point is assumed to be static.}
	\label{RVDseq}
\end{figure}

\subsection{Rendezvous Constraints}
In the rendezvous phase, collision avoidance constraints are typically modeled by enforcing a keep-out zone of radius $r$ (see Fig.~\ref{RVDseq}), i.e.,
\begin{equation}\label{koz}
\|\mathbf{x}_p(k)\|\geq r.
\end{equation}
The above constraint is nonconvex, and hence not compatible with the formulation \eqref{fpt2}-\eqref{ssconstr3} ($\mathcal{X}(k,N)$ is assumed to be convex). One way to convexify \eqref{koz} is to employ the so-called rotating-hyperplane method \cite{zagaris18}. Within this method, the hyperplane rotation rate is treated as a parameter to be tuned heuristically. The heuristic proposed herein amounts to parameterizing the hyperplane rotation in terms of both the time index $k$ and the horizon length $N$. In particular, the following safety constraint set is enforced:
\begin{equation}\label{hsetpar}
\mathcal{R}(k,\lambda_N)=\Big\{\boldsymbol{\xi}\in\mathbb{R}^3\,:\, \boldsymbol{\xi}^T\vec{\boldsymbol{\nu}}(k,\lambda_N)\geq r \Big\},
\end{equation}
which defines a half-space by means of a separating plane passing through the point $r \vec{\boldsymbol{\nu}}(k,\lambda_N)$ with outward unit normal $\vec{\boldsymbol{\nu}}(k,\lambda_N)$. The vector $\vec{\boldsymbol{\nu}}(k,\lambda_N)$ is obtained from the geodesic joining the projections on the unit sphere of the initial position ${\mathbf{x}}_p(k_0)$ and the reference position at the end of the rendezvous phase ${\mathbf{x}}_p^{d}(k_0 +\lambda_N)$, as
\begin{equation}\label{nukd1}
\begin{array}{l l l}
\vec{\boldsymbol{\nu}}(k,\lambda_N)&=&\mathbf{R}\big(\mathbf{a}_x(\lambda_N),\theta(k,\lambda_N)\big)\,\vec{\mathbf{x}}_p(k_0),
\end{array}
\end{equation}
where
\begin{equation}\label{nukd2}
\mathbf{a}_x(\lambda_N)=\vec{\mathbf{x}}_p(k_0)\times \vec{\mathbf{x}}_p^{\,d}(k_0 +\lambda_N)
\end{equation}
and
\begin{equation}\label{nukd3}
{\theta}(k,\lambda_N)=\dfrac{k-k_0}{\lambda_N}\arccos(\vec{\mathbf{x}}_p^{\,T}\!(k_0)\,\vec{\mathbf{x}}_p^{\,d}(k_0 +\lambda_N)).
\end{equation}

The set $\mathcal{R}(k,\lambda_N)$ in \eqref{hsetpar}-\eqref{nukd3} has the property that the constraint $\mathbf{x}_p(k)\in \mathcal{R}(k,\lambda_N)$ implies \eqref{koz}.  With respect to \cite{weiss15,zagaris18}, the half-space $\mathcal{R}(k,\lambda_N)$ has been designed in such a way that its gradual rotation favours the transition from the initial position ${\mathbf{x}}_p(k_0)$ to the reference position ${\mathbf{x}}_p^{d}(k_0 +\lambda_N)$, where the docking phase will start. Notice that the time scale of such rotation, and consequently the conservativeness of the keep-out-zone approximation, is affected by the horizon-dependent variable $\lambda_N$. For any fixed $\mathbf{x}_p^d$, the approximation gets better for increasing values of $\lambda_N$ (see Fig.~\ref{conill}), and thus of $N$ (see \eqref{lambdaN}).
	
	\begin{figure}[!t]
		\centering
		\psfrag{a}{\small$\lambda_N=1$}
		\psfrag{b}{\small$\vec{\mathbf{x}}_p(k_0)$}
		\psfrag{c}{{\color{red}\small$\mathcal{R}(k_0,1)$}}
	    \psfrag{k}{{\color{red}\small$\mathcal{R}(k_0,2)$}}
		\psfrag{d}{{\color{red}\small$\mathcal{R}(k_0+1,1)$}}
		\psfrag{e}{\small$\vec{\mathbf{x}}_p^{\,d}(k_0 +1)$}
		\psfrag{f}{\small$\lambda_N=2$}
		\psfrag{g}{{\color{red}\small$\mathcal{R}(k_0+1,2)$}}
		\psfrag{h}{{\color{red}\small$\mathcal{R}(k_0+2,2)$}}
		\psfrag{i}{\small$\vec{\mathbf{x}}_p^{\,d}(k_0 +2)$}
		\includegraphics[width=0.8\textwidth]{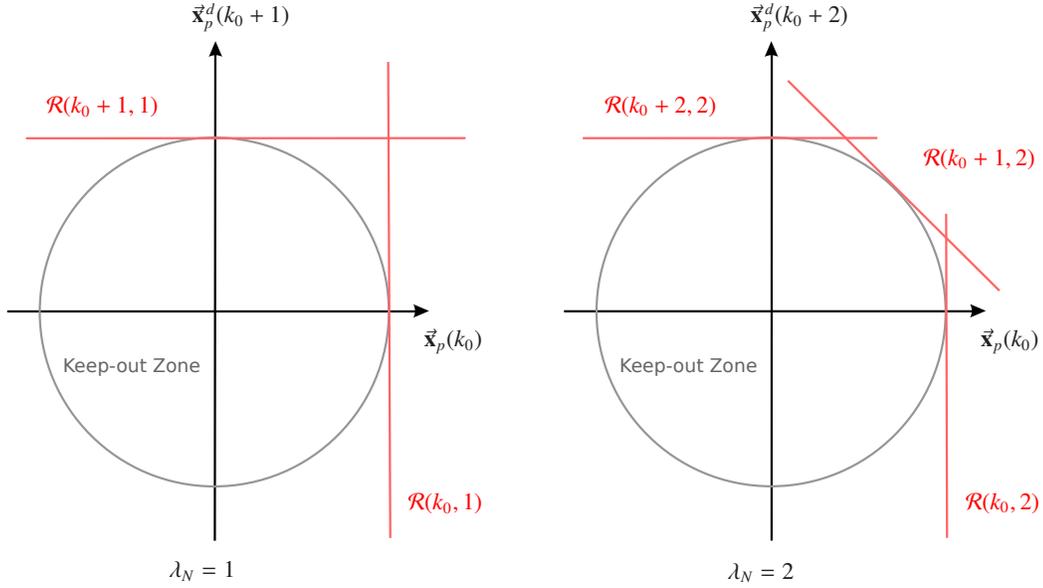}
		\caption{Illustration of the keep-out-zone approximation scheme on a two-dimensional example, for different values of $\lambda_N$.}
		\label{conill}
	\end{figure}

\subsection{Docking Constraints}
In order to define the constraints for the docking phase, let us consider the time varying set
\begin{equation}\label{dockcorr}
\mathcal{C}(k)=\Bigg\{\boldsymbol{\xi}\in\mathbb{R}^3\,:\, \left\|\boldsymbol{\xi}-[\boldsymbol{\xi}^T \vec{\mathbf{x}}_p^{\,d}(k)]{\vec{\mathbf{x}}_p^{\,d}}(k)\right\|_2 \leq \tan(\alpha) [\boldsymbol{\xi}-\mathbf{x}_p^d(k)]^T \vec{\mathbf{x}}_p^{\,d}(k) \Bigg\},
\end{equation}
which consists of a cone stemming from the docking point $\mathbf{x}_p^d(k)$, whose half-angle amplitude is defined by the parameter $\alpha$ (see again Fig.~\ref{RVDseq}). The visibility condition during docking can be defined as $\mathbf{x}_p(k)\in \mathcal{C}(k)$, which represents a convex quadratic constraint. For the sake of computational performance, however, it is convenient to devise a polyhedral constraint set implying $\mathbf{x}_p(k)\in \mathcal{C}(k)$. In order to accomplish this, we exploit the inequality
\begin{equation}\label{lineq}
\|\mathbf{y}\|_2 =\| \mathbf{T} \mathbf{y}\|_2\leq \sqrt{\rho}\| \mathbf{T} \mathbf{y}\|_\infty ,
\end{equation}
which holds for any $\mathbf{y}\in\mathbb{R}^3$ and any orthogonal matrix $\mathbf{T}\in\mathbb{R}^{3\times 3}$, being $\rho$ the number of nonzero elements in $\mathbf{T}\mathbf{y}$.
Let us consider the (orthogonal) rotation matrix $\mathbf{T}(k)$ which takes ${\mathbf{x}}_p^{\,d}(k)$ to a basis vector of the RTN frame. Then, $\mathbf{T}(k)\{\boldsymbol{\xi}-[\boldsymbol{\xi}^T \vec{\mathbf{x}}_p^{\,d}(k)]{\vec{\mathbf{x}}_p^{\,d}}(k)\}$ has at most two nonzero elements, since $\boldsymbol{\xi}-[\boldsymbol{\xi}^T \vec{\mathbf{x}}_p^{\,d}(k)]{\vec{\mathbf{x}}_p^{\,d}}(k)$ and $\vec{\mathbf{x}}_p^{\,d}(k)$ are orthogonal.
Therefore, introducing the polyhedral set
\begin{equation}\label{dockcorrpoly}
\mathcal{D}(k)=\Bigg\{\boldsymbol{\xi}\in\mathbb{R}^3\,:\, \left\|\mathbf{T}(k)\left( \boldsymbol{\xi}-[\boldsymbol{\xi}^T \vec{\mathbf{x}}_p^{\,d}(k)]{\vec{\mathbf{x}}_p^{\,d}}(k)\right)\right\|_\infty \leq \dfrac{\tan(\alpha)}{\sqrt{2}} [\boldsymbol{\xi}-\mathbf{x}_p^d(k)]^T \vec{\mathbf{x}}_p^{\,d}(k) \Bigg\},
\end{equation}
we get from \eqref{lineq} that $\mathcal{D}(k) \subseteq \mathcal{C}(k)$,  as illustrated in Fig.~\ref{dckcon}.
Hence, the linear constraint set $\mathbf{x}_p(k)\in\mathcal{D}(k)$ implies the visibility condition $\mathbf{x}_p(k)\in\mathcal{C}(k)$.
In this work, the matrix $\mathbf{T}(k)$ is chosen as
\begin{equation}\label{tmat}
\mathbf{T}(k)=\mathbf{R}\Big(\vec{\mathbf{x}}_p^{\,d}(k)\times [1\; 0\; 0]^T,\,\arccos([1\; 0\; 0]\vec{\mathbf{x}}_p^{\,d}(k))\Big) .
\end{equation}

Albeit not discussed here for conciseness, other types of maneuver constraints naturally fit the proposed formulation. For instance, plume-impingement constraints can be accommodated for by enforcing a time-varying bound on the control input magnitude, see \cite{breger08safe}. Note that plume impingement is less likely to occur when the docking point is spinning, as the servicer thrust is mostly directed away from the target during the final approach, in order to compensate for centrifugal effects.

	\begin{figure}[t]
		\centering
		\psfrag{a}{{\color{blue}\small$\mathcal{C}(k)$}}
		\psfrag{b}{{\color{red}\small$\mathcal{D}(k)$}}
		\psfrag{T}{\small $\mathbf{x}_p^d(k)$}
		\includegraphics[width=0.35\textwidth]{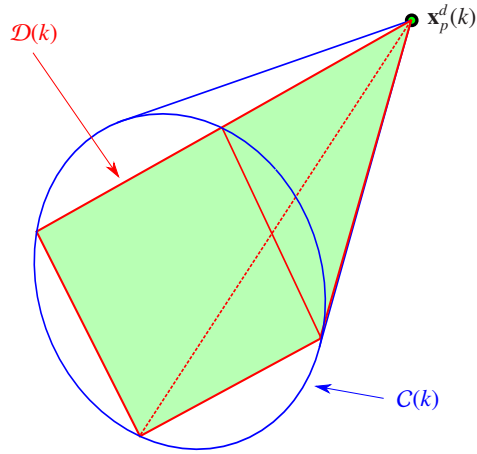}
		\caption{Illustration of the polyhedral set $\mathcal{D}(k)$ used to approximate the docking cone $\mathcal{C}(k)$.}
		\label{dckcon}
	\end{figure}

\section{Solution Strategy}\label{solution}
Problem~\eqref{fpt2} is nonconvex due to the integer optimization variable $N$. Nevertheless, it can readily be tackled by solving a sequence of linear programming problems of the form
\begin{equation}\label{fpt3}
\begin{aligned}
\underset{\mathbf{u}_N}{\text{min}}~~ &{J}_{N}= {N}+ \gamma\|\mathbf{u}_{N}\|_1 \\[1mm]
\text{s.t.} \quad & {\mathbf{x}}(k_0)= \mathbf{x}_0\\[1mm]
& {\mathbf{x}}(k+1)=\mathbf{A} {\mathbf{x}}(k)+\mathbf{B}{\mathbf{u}}(k)  \\[1mm]
&  {\mathbf{x}}_p(k)\in \mathcal{X}(k,N) \quad k=k_0,\ldots, k_0+{N}-1\\[1mm]
& \|\mathbf{u}_{N}\|_\infty \leq 1 \\[1mm]
&  {\mathbf{x}}(k_0+N)=\mathbf{x}^d(k_0+N)
\end{aligned}
\end{equation}
where the planning horizon ${N}$ is now a fixed parameter. We denote the minimizer of \eqref{fpt3} by $\mathbf{u}_{N}^*$, and the corresponding optimal cost by ${J}^*_{N}$. If the problem is infeasible, the cost is set by definition to ${J}^*_{N}=\infty$. In order to find an optimal solution of problem \eqref{fpt2}, one can solve \eqref{fpt3} to obtain the function ${J}^*_{N}: \mathbb{N}\rightarrow \mathbb{R}^+$, and then minimize ${J}^*_{N}$ with respect to $N\in\mathcal{I}=\{1,\ldots,N_{ub}\}$, where $N_{ub}$ is an upper bound on the optimal horizon $N^*$ of \eqref{fpt2}. Being ${J}_{N}$ unbounded for $N\rightarrow\infty$, the upper bound $N_{ub}$ is guaranteed to be finite.

The profile of ${J}^*_{N}$ versus $N$ is reported in Fig.~\ref{cfprof} for an example of RVD maneuver. It can be seen that the horizon length $N$ has a profound impact on maneuver performance, and that optimizing over $N$ is a inherently nonconvex problem.
\begin{figure}[t]
	\centering
	\includegraphics[width=0.65\textwidth]{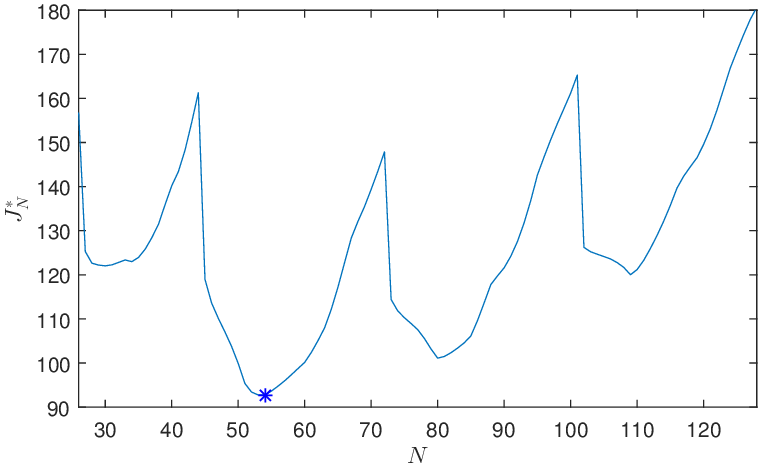}
	\caption{Profile of ${J}^*_{N}$ versus $N$ for the RVD scenario detailed in Section \ref{SAp}, with $\gamma=4$. The presence of multiple local minima is due to the target rotation.  The global optimum $N^*$ is marked by an asterisk. The problem is infeasible for $N\leq 25$.}
	\label{cfprof}
\end{figure}
Hence, searching for the global optimum ${N^*}$ turns out to be as hard as solving \eqref{fpt3} for all $N\in\mathcal{I}$. In order to rule out some of the values of $N$, the following result can be employed.
\begin{proposition}\label{pfeas}
Consider the linear system $\mathbf{x}(k+1)=\mathbf{A} \mathbf{x}(k)+ \mathbf{B} \mathbf{\mathbf{u}}(k)$, where $\mathbf{x}(k)\in\mathbb{R}^n$ and $\mathbf{u}(k)\in\mathbb{R}^m$, and let $\mathbf{x}^d(k_0+N)\in\mathbb{R}^n$ be a target state to be reached from the initial state $\mathbf{x}(k_0)=\mathbf{x}_0$. Define the $N$-step reachability matrix $\mathbf{R}_N=[\mathbf{B}\, \mathbf{A}\mathbf{B}\, \ldots \mathbf{A}^{N-1}\mathbf{B}]$ and the unconstrained minimum energy (least-squares) input sequence $\mathbf{e}_{N}=\mathbf{R}_N^\dag(\mathbf{x}^d(k_0+N)-\mathbf{A}^N \mathbf{x}_0)$. Moreover, let $\mathcal{I}\subseteq\mathbb{N}$ and
\begin{equation}\label{Jset}
\mathcal{F}=\left\{ N\in\mathcal{I}:\,
\begin{array}{r c l}
\mathbf{x}^d(k_0+N)&=&\mathbf{A}^N \mathbf{x}_0+\mathbf{R}_N \mathbf{e}_{N},\\  \|\mathbf{e}_{N}\|_2&\leq& \sqrt{mN}
\end{array}
\right\}
\end{equation}
Then, the feasibility problem
\begin{equation}\label{fpt4}
\begin{aligned}
{\text{find}}~~ & \mathbf{u}_{N} \\[1mm]
\text{s.t.} \quad
& {\mathbf{x}}(k+1)=\mathbf{A} {\mathbf{x}}(k)+\mathbf{B}{\mathbf{u}}(k)  \\[1mm]
& {\mathbf{x}}(k_0)= \mathbf{x}_0,\;\;{\mathbf{x}}(k_0+N)=\mathbf{x}^d(k_0+N) \\[1mm]
& \|\mathbf{u}_{N}\|_\infty \leq 1
\end{aligned}
\end{equation}
has no solution for $N\in\{\mathcal{I} \setminus \mathcal{F}\}$.
\end{proposition}
\begin{proof}
The equality in \eqref{Jset} stems from the equality constraints of problem \eqref{fpt4}. In order to prove necessity of the inequality in \eqref{Jset}, observe that $\mathbf{e}_{N}$ is the minimum energy input sequence that drives the system from state $\mathbf{x}_0$ at time $k_0$ to state $\mathbf{x}^d$ at time $k_0+N$. Hence,
\begin{equation}\label{ineqchain}
\|\mathbf{e}_{N}\|_2\leq \|\mathbf{u}_{N}\|_2 \leq \sqrt{mN}\|\mathbf{u}_{N}\|_\infty.
\end{equation}
A necessary condition for problem \eqref{fpt4} to be feasible is that $\|\mathbf{u}_{N}\|_\infty \leq 1$. This can happen only if $\|\mathbf{e}_{N}\|_2 \leq\sqrt{mN}$.
\end{proof}

\begin{remark}\label{remNmin}
By applying Proposition \ref{pfeas} to system \eqref{zetacirctd}, it follows that problem
\eqref{fpt3} has no solution for $N\in\{\mathcal{I} \setminus \mathcal{F}\}$. Moreover, the set $\mathcal{I}$ in \eqref{Jset} is chosen as $\mathcal{I}=\{1,\ldots,N_{ub}\}$. Then, the smallest element of $\mathcal{F}$ is a lower bound on the minimum horizon $N=N_{min}\,$ for which problem \eqref{fpt3} is feasible.
\end{remark}

The minimum energy control for different values $N$ of the horizon length can be evaluated very efficiently, given that the reachability matrix $\mathbf{R}_N$, its pseudoinverse $\mathbf{R}_N^\dag$, and matrix $\mathbf{A}^N$ can be pre-computed offline. However, solving problem \eqref{fpt3} for all $N\in\mathcal{F}$ is still prohibitively complex when the number of elements in $\mathcal{F}$ is large. In order to mitigate this issue, one has to settle for a locally optimal solution. More specifically, our aim is to solve a small subset of problems \eqref{fpt3}, in which the choice of $N$ is guided by a local search within the set $\mathcal{F}$. To this purpose, the following three-step solution method is proposed:
\begin{algorithm}
\caption{Returns a local optimum $\boldhat{N}$ of problem \eqref{fpt2}.}
\label{alg1}
\begin{algorithmic}
   \State{let $\mathcal{F}(q)$ denote the $q$-th largest integer in the
set $\mathcal{F}$}

   \If{$\mathcal{F}=\varnothing$}
   \State \Return{problem \eqref{fpt2} infeasible}
   \Else

   \Comment{Step 1.}
   \State{let $N_1=\underset{N\in\mathcal{F}}{\text{argmin}}\big[N+
\gamma \| E_{N} \|_1\big]$}

   \Comment{Step 2.}
   \State{let $q_1$ be such that $\mathcal{F}(q_1)=N_1$, let $i=0$}
   \Repeat
   \State{solve problem \eqref{fpt3} with $N=\mathcal{F}(q_1+i)$ and
$N=\mathcal{F}(q_1-i), \quad i=i+1$}
   \Until{a feasible solution is found or both the endpoints of
$\mathcal{F}$ are reached}
   \If{a feasible solution is found}
   \State{let $N_2$ be the horizon length}
   \ElsIf{two feasible solutions are found at $i=\zeta$}
   \If{${J}^*_{\mathcal{F}(q_1+\zeta)}<{J}^*_{\mathcal{F}(q_1-\zeta)}$}
   \State{let $N_2=\mathcal{F}(q_1+\zeta)$}
   \Else
   \State{let $N_2=\mathcal{F}(q_1-\zeta)$}
   \EndIf
   \Else
   \State \Return{problem \eqref{fpt2} infeasible}
   \EndIf

   \Comment{Step 3.}
   \State{let $q_2$ be such that $\mathcal{F}(q_2)=N_2$, let $i=1$}
   \If{$N_2>N_1$}
   \Repeat
   \State{solve problem \eqref{fpt3} with $N=\mathcal{F}(q_2+i), \quad
i=i+1$}
   \Until{an index $\varsigma$ is found such that
${J}^*_{\mathcal{F}(q_2+\varsigma)}>{J}^*_{\mathcal{F}(q_2+\varsigma-1)}$}
   \State{let $\boldhat{N}=\mathcal{F}(q_2+\varsigma-1)$}
   \State \Return{$\boldhat{N},\ {J}^*_{\mathcal{F}(q_2+\varsigma-1)}$ }
   \ElsIf{$N_2<N_1$}
   \Repeat
   \State{solve problem \eqref{fpt3} with $N=\mathcal{F}(q_2-i),\quad
i=i+1$}
   \Until{an index $i=\varsigma$ is found such that
${J}^*_{\mathcal{F}(q_2-\varsigma)}>{J}^*_{\mathcal{F}(q_2-\varsigma+1)}$}
   \State{let $\boldhat{N}=\mathcal{F}(q_2-\varsigma+1)$}
   \State \Return{$\boldhat{N},\ {J}^*_{\mathcal{F}(q_2-\varsigma+1)}$ }
   \Else
   \State{let $\underline{J} = \text{min} \{{J}^*_{\mathcal{F}(q_1-1)},\
{J}^*_{\mathcal{F}(q_1)},\ {J}^*_{\mathcal{F}(q_1+1)}\}$}
   \If{$\underline{J} = {J}^*_{\mathcal{F}(q_1-1)}$}
   \State{proceed as for $N_2< N_1$}
   \ElsIf{$\underline{J} = {J}^*_{\mathcal{F}(q_1+1)}$}
   \State{proceed as for $N_2>N_1$}
   \Else
   \State{let $\boldhat{N}=N_1$}
   \State \Return{$\boldhat{N},\ {J}^*_{\mathcal{F}(q_1)}$ }
   \EndIf

   \EndIf

   \EndIf
\end{algorithmic}
\end{algorithm}
\begin{enumerate}
\item Compute an initial guess $N_1$ of the horizon length by evaluating the cost associated with the minimum energy control $\mathbf{e}_{N}$, i.e.,
\begin{equation}\label{initguess}
  N_1=\underset{N\in\mathcal{F}}{\text{argmin}}\big[N+ \gamma \| \mathbf{e}_{N} \|_1\big].
\end{equation}
If multiple minima are found, take the one with the smallest $N$ (among equivalent solutions, the one featuring the smallest $N$ is preferred). Notice that $N=N_1$ is not guaranteed to be feasible for problem \eqref{fpt3}.
\item Starting from $N=N_1$, perform a local search within the set $\mathcal{F}$ until a feasible solution to problem \eqref{fpt3} is found (see Algorithm~\ref{alg1}). This step amounts to solving a sequence of problems \eqref{fpt3} with different values of $N$ in a neighborhood of $N_1$. If a feasible solution is found, denote by $N_2$ the corresponding horizon length. If no feasible solution is found for any $N\in\mathcal{F}$, then mark problem \eqref{fpt2} as infeasible.
\item If $N_2>N_1$, solve a sequence of problems \eqref{fpt3} with increasing horizon length within the set $\mathcal{F}$, starting from $N=N_2$. Stop when the optimal cost of \eqref{fpt3} does not decrease anymore. Similarly, if $N_2<N_1$, apply the same procedure but with decreasing horizon length. If $N_2=N_1$, decide whether to increase or decrease the horizon length by comparing the cost ${J}^*_{N_1}$ with two neighbouring solutions of \eqref{fpt3} within the set $\mathcal{F}$. The optimized horizon length resulting from this step is denoted by $\boldhat{N}$.
\end{enumerate}
Algorithm~\ref{alg1} formalizes the proposed method.

It is worth remarking that Steps~1-3 of Algorithm~\ref{alg1} are motivated by the practical need to trade off performance and computational efficiency. In Step~1, the initial guess $N_1$  is obtained based on the observation that in the domain $\mathcal{F}$ the profile of $N+\gamma\| \mathbf{e}_{N}\|_1$ is often close to that of $N+\gamma\| \mathbf{u}^*_{N}\|_1$, provided by \eqref{fpt3}. In Steps~2-3, a local search is performed in the neighborhood of $N_1$. By virtue of Step 2 a feasible solution of problem \eqref{fpt2} is always found, if it exists. Moreover, according to Step 3, convergence to a local minimum $\boldhat{N}$ of ${J}^*_{N}$ is guaranteed. Global optimality (i.e., $\boldhat{N}=N^*$) is ensured if one of the following conditions is met:
\begin{itemize}
\item[(i)] ${J}^*_{N}$ has no local minima in the domain $\mathcal{F}$;
\item[(ii)] $\gamma=0$, in which case $N_1$ is a lower bound on $N_{min}$ (see \eqref{initguess} and Remark \ref{remNmin}), and Steps $2$-$3$ are guaranteed to find $\boldhat{N}=N^*=N_{min}$.
\end{itemize}
The performance of Algorithm~\ref{alg1} is evaluated in the next section.

\section{Search Algorithm Validation}\label{SAp}
A qualitative assessment of the proposed solution strategy has been carried out by testing Algorithm~\ref{alg1} on a specific RVD scenario, for different values of the parameter $\gamma$ in \eqref{fpt3}. The motivation of this parametric study is that larger values of $\gamma$ usually correspond to longer planning horizons and thus to an increased computational load. This allows one to draw some conclusions regarding the computational feasibility of the method. The RVD maneuver parameters are summarized in Table~\ref{rvdtab}, where $t_0$ denotes the initial time. They are consistent with the specifications of a small satellite mission tailored to the removal of a debris object in Low-Earth-Orbit, see, e.g., \cite{leomanni20}.
\begin{table}[b]
\centering\caption{RVD test maneuver parameters}\vspace{3mm}
\begin{tabular}{ l l }
\hline \hline
Parameter & Value\\ \hline
Initial docking point position & $\mathbf{p}^d(t_0) = [1,\,0,\,0]^T$ m\\ 
Target angular velocity &  $\boldsymbol{\omega}(t) = [0,\,0,\,0.01]^T$ rad/s\\ 
Target mean motion & $\eta = 0.001$ rad/s\\  
Servicer maximum acceleration & $a_{max} = 0.001$ m/s$^2$\\  
Sampling interval  & $\tau_s = 2\pi/256$ rad/sample \\
Initial relative position &  $\frac{a_{max}}{\eta^2}\mathbf{x}_p(k_0) = [0,\,-100,\,0]^T$ m\\ 
Initial relative velocity &  $\frac{a_{max}}{\eta}\mathbf{x}_v(k_0) = [0,\,0,\,0]^T$ m/s\\ 
Docking cone half-angle & $\alpha = 20$ deg\\  
Keep-out zone radius & $\frac{a_{max}}{\eta^2}r = 5$ m  \\  
Docking phase duration &$\frac{\tau_s}{\eta} N_d = 220.9$ s $(N_d=9)$ \\  \hline\hline
\end{tabular}\label{rvdtab}
\end{table}

Within this setup, the performance of Algorithm~\ref{alg1} has been compared with that of full enumeration and binary search methods. The full enumeration approach amounts to solving problem \eqref{fpt3} for all $N\in\mathcal{I}$, and it is guaranteed to find the global optimum $N^*$ of problem \eqref{fpt2}. The binary search method minimizes $J_N^*$ with respect to $N$ by bisection on the interval $\mathcal{I}$. Its complexity is logarithmic in $N_{ub}$, and in general it returns a local minimum of $J_N^*$ . In this study, the horizon upper bound is set to $N_{ub}=128$ samples.

\begin{figure}[!t]
	\centering
	\includegraphics[width=0.65\textwidth]{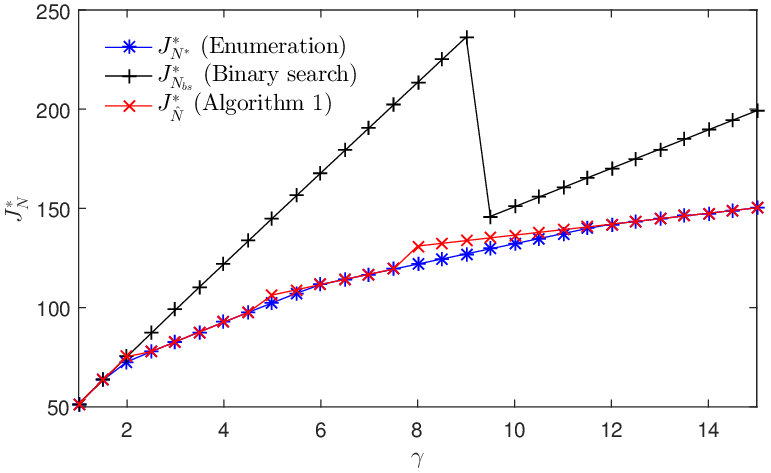}
	\caption{Maneuver cost incurred by the considered search strategies for different values of $\gamma$.}
	\label{fv_gamma}
\end{figure}
\begin{figure}[!t]
	\centering
	\includegraphics[width=0.65\textwidth]{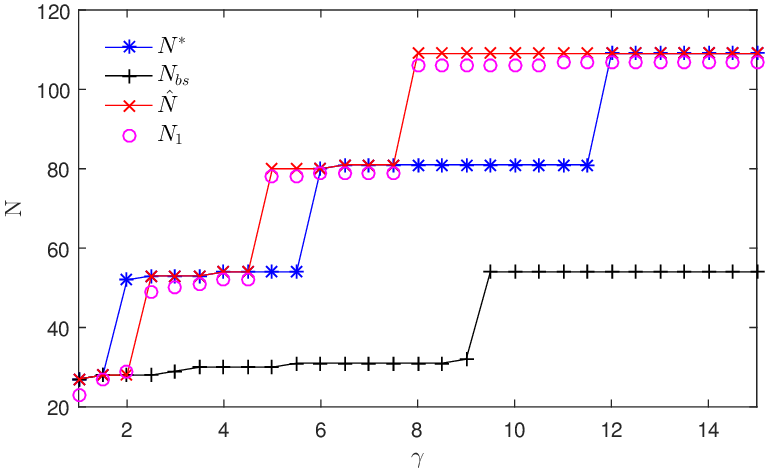}
	\caption{Optimized planning horizon provided by the considered search strategies for different values of $\gamma$. The initial guess $N_1$ of Algorithm~\ref{alg1} is also reported.}
	\label{N_gamma}
\end{figure}
\begin{figure}[!t]
	\centering
	\includegraphics[width=0.65\textwidth]{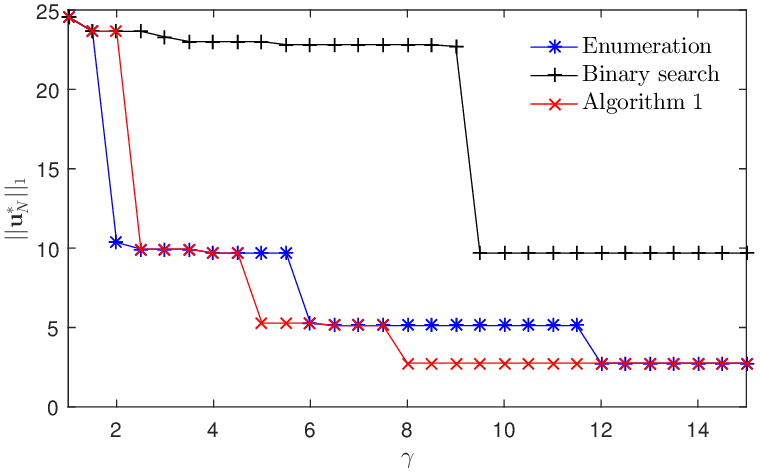}
	\caption{Fuel consumption incurred by the considered search strategies for different values of $\gamma$.}
	\label{u1_gamma}
\end{figure}

The obtained results are reported in Figs.~\ref{fv_gamma}-\ref{u1_gamma} for different values of $\gamma$ in the range $\gamma\in[1,\; 15]$. Figure~\ref{fv_gamma} shows the value function $J^*_N$, evaluated at the horizon length returned by the three compared methods.  It can be seen that the cost $J^*_{\boldhat{N}}$ incurred by Algorithm~\ref{alg1} is either equal or very close to the optimal cost $J^*_{N^*}$ obtained via full enumeration. Conversely, the cost $J^*_{N_{bs}}$ provided by the binary search algorithm is always far from global optimality except for small values of $\gamma$. This is not surprising, since the binary search method is known to work well for the minimization of unimodal functions, while the unimodality condition is not met for the problem at hand (this is evident in Fig.~\ref{cfprof}, which depicts the profile of $J^*_{N}$ corresponding to the parameters in Table~\ref{rvdtab}, for $\gamma=4$). Figures~\ref{N_gamma}~and~\ref{u1_gamma} report, respectively, the optimized horizon length and the normalized fuel consumption $\|\mathbf{u}^*_N\|_1$ incurred by each solution strategy. It can be seen that the binary search solution $N_{bs}$ tends to underestimate $N^*$, while requiring a much higher fuel consumption with respect to the other approaches. Conversely, the horizon length $\boldhat{N}$ provided by Algorithm~\ref{alg1} is very close to $N^*$ in many instances. When $\boldhat{N}> N^*$, the corresponding fuel cost is such that $\|\mathbf{u}^*_{\boldhat{N}}\|_1<\| \mathbf{u}^*_{N^*}\|_1$. Figure~\ref{N_gamma} also depicts the initial guess $N_1$ defined by \eqref{initguess}. It can be seen that the initial guess often falls reasonably close to a local optimum. Table~\ref{a1tab} reports the minimum, maximum, and average CPU time over the considered values of $\gamma$ for the binary search and Algorithm~\ref{alg1} solutions, showing that the computational burden of the two methods is on a similar level.
\begin{table}[!t]
\centering\caption{Computational burden of binary search and Algorithm~\ref{alg1}}\vspace{3mm}
\begin{tabular}{l c c}
\hline \hline
CPU time & Binary search & Algorithm~\ref{alg1}\\ \hline
Minimum & 0.069 s& 0.020 s \\
Average & 0.094 s & 0.104 s\\
Maximum & 0.126 s &0.165 s\\
\hline\hline
\end{tabular}\label{a1tab}
\end{table}
This is a remarkable result, given that $\boldhat{N}\gg N_{bs}$ in most problem instances (see again Fig.~\ref{N_gamma}). The CPU time of the full enumeration procedure is by far higher than that of these two approaches and amounts to approximately 3 s, regardless of the value of $\gamma$. From these results, it can be concluded that the proposed solution strategy provides an excellent trade-off between performance and computational efficiency.

Finally, it is worth noticing that variable-horizon problems can often be cast as a MILP, see, e.g., \cite{richards2006robust}. A MILP formulation of problem \eqref{fpt2} has been tested, but the obtained results turned out to be unsatisfactory. In part, this is due to the fact that the feasible region~\eqref{ssconstr3} is parameterized by an explicit function of $N$. In order to cope with this issue, one has to construct a MILP including all realizations of $\mathcal{X}(k,N)$, obtained for $N\in\mathcal{I}$ and $k=k_0,\ldots, k_0+{N}-1$, resulting in a number of state constraints which is proportional to $N_{ub}^2$. Hence, the problem rapidly becomes untractable as $N_{ub}$ grows. In the considered scenario, even by removing the state constraints, the MILP solution time is in the order of seconds. All computations have been performed on a standard laptop, via a direct call from Matlab of the commercial solver Gurobi \cite{gurobi}.
\section{Rendezvous and Docking to the EnviSat Platform}\label{Ervd}

In recent years, a number of studies have focused on in-orbit servicing missions dedicated to capturing and de-orbiting the European Space Agency (ESA) EnviSat platform, see, e.g., \cite{estable20}. The EnviSat operational life ended on April 8, 2012, following the unexpected loss of contact with the spacecraft. After this event, the spacecraft lost the ability to hold its Earth-pointing attitude and started to tumble. Due to its huge size and its particular orbital configuration, EnviSat is currently regarded as a potential trigger for space debris proliferation in low Earth orbit. In the following, the proposed guidance scheme is demonstrated on an RVD scenario inspired by the capture of EnviSat.

\begin{figure}[!b]
	\centering\vspace{3mm}
	\includegraphics[width=0.65\textwidth]{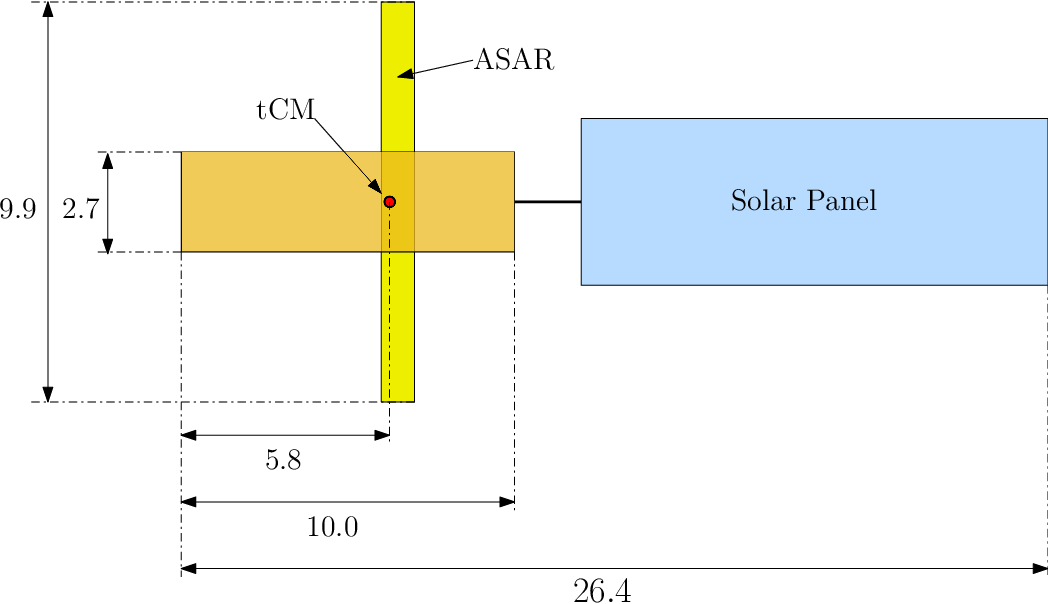}
	\caption{EnviSat spacecraft layout. The length of the various components is reported in meters.}
	\label{envisat}
\end{figure}
\begin{figure}[tb]
	\centering
	\includegraphics[width=0.4\textwidth]{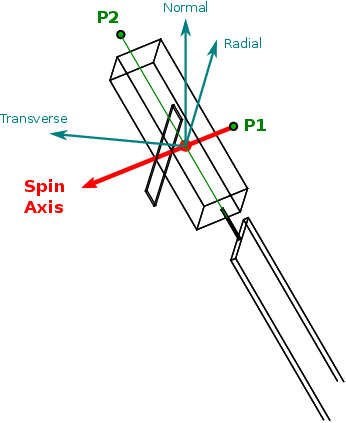}
	\caption{Characterization of the docking points (P1, P2) and of the spin axis in the RTN frame at the initial time $t_0$.}
	\label{after_tumb}
\end{figure}
\begin{table}[!b]
	\centering\caption{EnviSat RVD maneuver parameters}\vspace{3mm}
	\begin{tabular}{ l l }
		\hline \hline
		Parameter & Value\\ \hline
		\multirow{2}{*}{Initial docking point position} & P1: $\mathbf{p}^d(t_0) = [-0.0360,\,-2.6451,\,1.4149]^T$ m\\ & P2: $\mathbf{p}^d(t_0) = [-0.1683,\,3.5384,\,6.6107]^T$ m\\ 
		Initial angular velocity &  $\boldsymbol{\omega}(t_0) = [0.0003,\,0.0252,\,-0.0145]^T$ rad/s\\ 
		Target mean motion & $\eta = $ 0.001045 rad/s\\
		Servicer maximum acceleration & $a_{max} = 0.005$ m/s$^2$\\  
		Sampling interval  & $\tau_s = 2\pi/512$ rad/sample \\
		Initial relative position &  $\frac{a_{max}}{\eta^2}\mathbf{x}_p(k_0) = [0,\,-200,\,0]^T$ m\\ 
		Initial relative velocity &  $\frac{a_{max}}{\eta}\mathbf{x}_v(k_0) = [0,\,0,\,0]^T$ m/s\\ 
		Docking cone half-angle & $\alpha = 20$ deg\\  
		Keep-out zone radius & $\frac{a_{max}}{\eta^2}r = 22$ m  \\  
		Docking phase duration &$\frac{\tau_s}{\eta} N_d = 187.8$ s $(N_d=16)$ \\
		Weighting parameter $\gamma$ & $\gamma=4$\\ \hline\hline
	\end{tabular}\label{rvdtabenvi}
\end{table}

A schematic view of EnviSat is reported in Figure \ref{envisat}. The spacecraft was not designed with servicing in mind, and features an elongated shape with many protruding elements. Consequently, the determination of a suitable docking configuration is nontrivial. Following a review of the literature available on the topic (see, e.g., \cite{Deloo15,li2020real}), two favourable docking points have been identified: the first (P1) is located above the center of mass (tCM), in the direction opposite to the ASAR antenna; the second (P2) lies along the spacecraft long axis, in the direction opposite to the solar panel, as depicted in Fig.~\ref{after_tumb}. It is worth recalling that these points describe the desired position of the sCM upon docking (see Fig.~\ref{dkp}). In order to account for the geometrical configuration of the servicer, a clearance of 1.5 m is left between the docking points and the nearby EnviSat surfaces, similarly to what done in \cite{Deloo15}. Another important modeling issue is the characterization of the spin axis. It is generally agreed (see, e.g.,  \cite{kucharski2014attitude}) that the EnviSat spin axis is approximately fixed with respect to the body frame, and aligned with the vector joining P1 and the tCM, as depicted in Fig.~\ref{after_tumb}. Over the RVD maneuver time scale, one can safely assume that the spin axis is also inertially fixed. Under this assumption, the evolution of the angular velocity vector $\boldsymbol{\omega}(t)$ in \eqref{eq2} is given by
\begin{equation}\label{angvel}
\boldsymbol{\omega}(t)=
\left[
\begin{array}{c c c}
\;\cos(\eta (t-t_0))&\sin(\eta (t-t_0))&0 \\
-\sin(\eta (t-t_0))&\cos(\eta (t-t_0))&0 \\
0&0&1 \\
\end{array}
\right]\boldsymbol{\omega}(t_0) ,
\end{equation}
where $\boldsymbol{\omega}(t_0)$ is the angular velocity of the target body frame relative to the RTN frame at $t_0$. Based on the results in \cite{kucharski2014attitude}, the EnviSat spin period is taken as $220$ s, corresponding to an angular rate of $\| \boldsymbol{\omega}(t)\|_2= 0.029$ rad/s. The parameters chosen for the RVD maneuver simulations, reported in Table \ref{rvdtabenvi}, are consistent with the above discussion. Some comments about the selected value $a_{max} = 0.005$ m/s$^2$ of the servicer maximum acceleration are in order. For the problem at hand, reasonable values of $a_{max}$ may range from $10^{-3}$ m/s$^2$ to $10^{-1}$ m/s$^2$, on a rough estimate. The selected value is on the lower end of such interval. This choice is made in order to showcase the proposed method on a challenging optimization problem involving a low control authority. Moreover, the resulting maneuver plan may be realized by using small thrusters, which are lighter and more accurate than larger ones.

\begin{figure}[!t]
	\centering
	\includegraphics[width=0.75\textwidth]{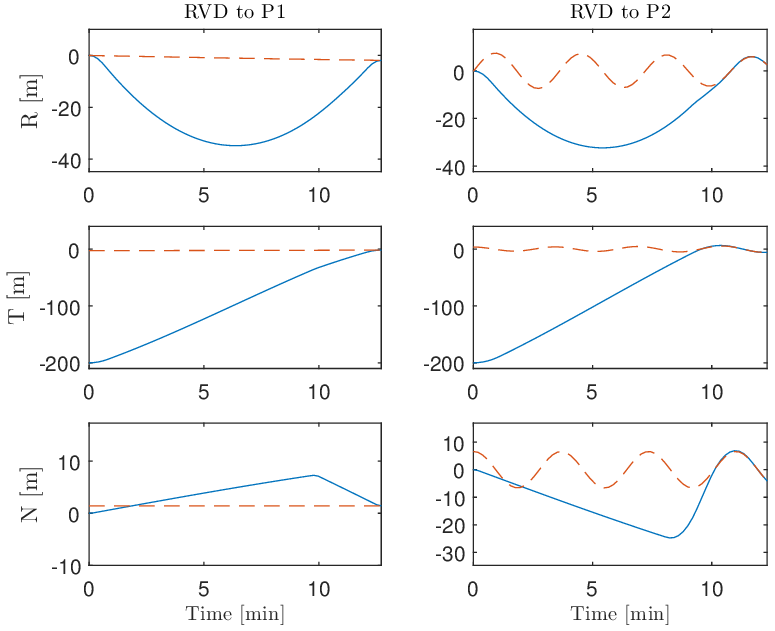}
	\caption{Radial (R), transverse (T) and normal (N) components of the relative position vector for RVD to P1 (left) and P2 (right): actual trajectory (solid) and reference trajectory (dashed).}
	\label{p_envi}
\end{figure}
\begin{figure}[!t]
	\centering
	\includegraphics[width=0.75\textwidth]{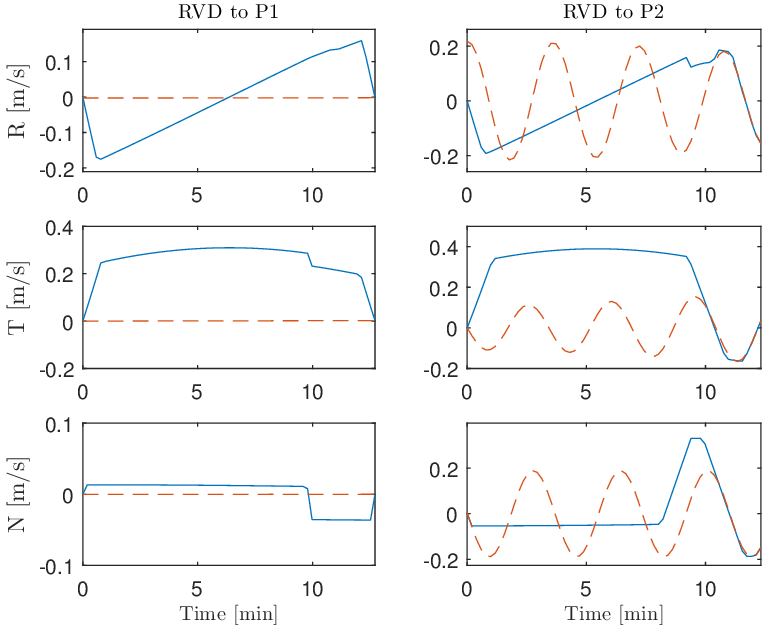}
	\caption{Radial (R), transverse (T) and normal (N) components of the relative velocity vector for RVD to P1 (left) and P2 (right): actual trajectory (solid) and reference trajectory (dashed).}
	\label{v_envi}
\end{figure}
\begin{figure}[!t]
	\centering
	\includegraphics[width=0.75\textwidth]{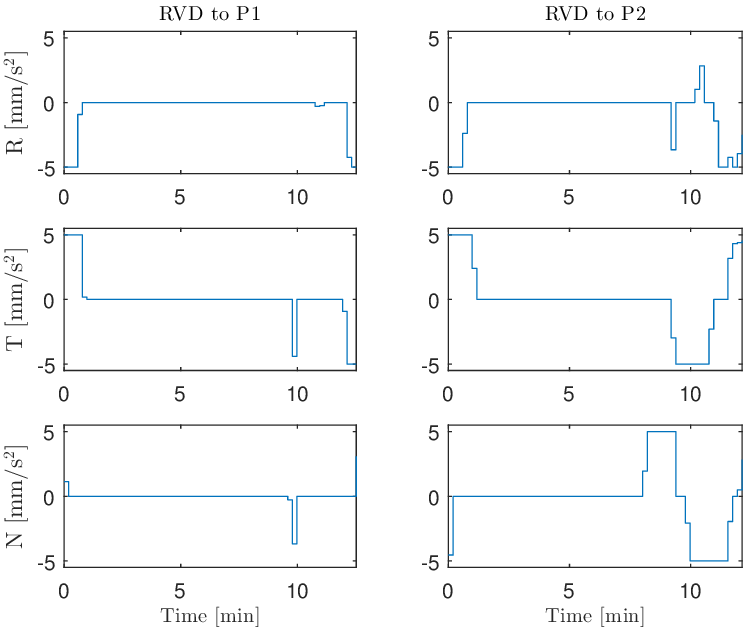}
	\caption{Radial (R), transverse (T) and normal (N) components of the servicer acceleration vector for RVD to P1 (left) and P2 (right).}
	\label{a_envi}
\end{figure}

Problem \eqref{fpt2} has been solved using Algorithm~\ref{alg1} for RVD maneuvers aimed at reaching the docking points P1 and P2. The resulting state and input trajectories are reported in Figs.~\ref{p_envi}-\ref{a_envi}. In both cases, the servicer spacecraft is successfully steered from a hold point located 200 m behind the target (i.e., EnviSat) to the selected docking point, in a time interval of approximately 12 minutes. Notice from Figs.~\ref{p_envi}-\ref{v_envi} that the reference trajectory (dashed) of P2 displays a much faster variation compared to that of P1. This is not surprising since the point P1 lies along the spin axis (see Fig.~\ref{after_tumb}). Consequently, the motion of P1 (as seen from the RTN frame) is only due to the precession of this axis, which occurs at the orbital rate $\eta$ (see \eqref{angvel}). Conversely, the point P2 is orthogonal to the spin axis. Therefore, its evolution in terms of RTN coordinates stems from both the spin axis precession and the rotation of P2 about the spin axis, the latter of which occurs at a frequency much higher (approximately 30 times) than $\eta$. As a result, reaching P2 is far more challenging than reaching P1. In fact, the fuel consumption corresponding to the control input trajectories in Fig.~\ref{a_envi} is about 3 times higher in the P2 case than in the P1 test. In Fig.~\ref{a_envi}, it can also be seen that the magnitude of the acceleration components stays within the assigned bound $a_{max}$ (reported in Table~\ref{rvdtabenvi}), over the entire maneuvering interval. Moreover, the obtained trajectories satisfy by construction the rendezvous and docking constraints described in Section \ref{docking}. Figure \ref{RVDEnvi3Dtrans} depicts the transition from the rendezvous constraints to the docking constraints, showing that feasibility is retained during this event. Figure \ref{RVDEnvi3D} illustrates how the docking corridor rotates during the final part of the maneuver. It can be seen that the sCM always lies inside the set defined by \eqref{dockcorrpoly}.

The obtained trajectories have been compared to those resulting from the solution of a continuous-time version of problem \eqref{fpt2}, in which the final time is free and the state constraints are nonlinear. In particular, the keep-out zone is enforced as in \eqref{koz}, while the docking corridor is modeled as in \eqref{dockcorr}. The transition between the rendezvous and docking phases is accounted for by formulating a two-phase optimal control problem, which is solved by using the commercial package GPOPS-II \cite{patterson2014gpops}.
\begin{figure}[!t]
\centering
\includegraphics[width=0.66\textwidth]{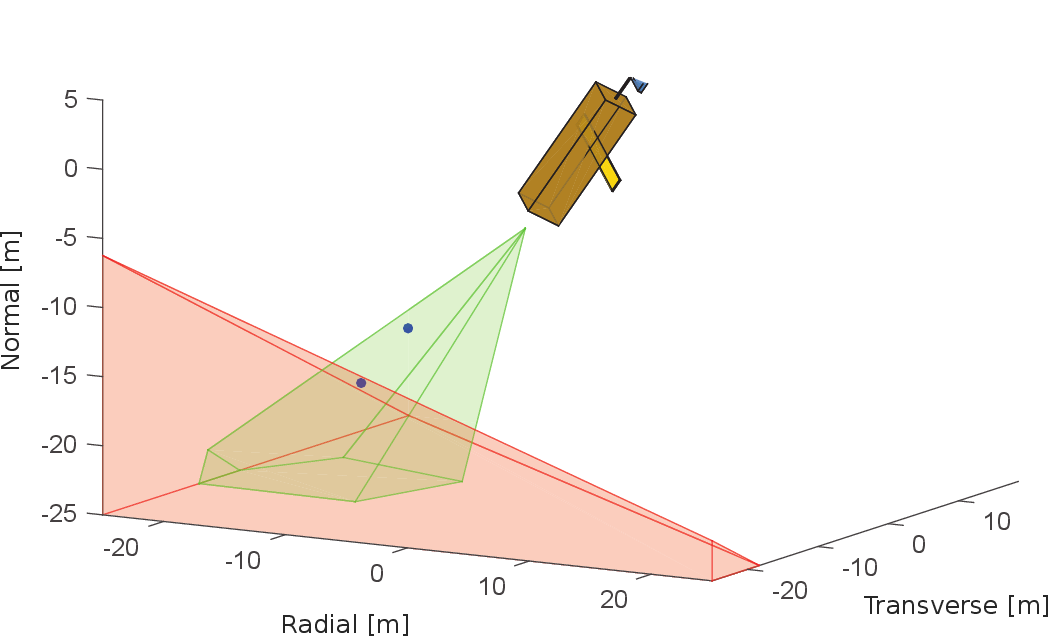}
\caption{Illustration of the transition from the safe rendezvous region (red) to the docking corridor (green), for RVD to P2. The servicer center of mass (sCM) is marked by a blue point.}
\label{RVDEnvi3Dtrans}
\end{figure}
\begin{figure}[!t]
	\centering
	\includegraphics[width=0.66\textwidth]{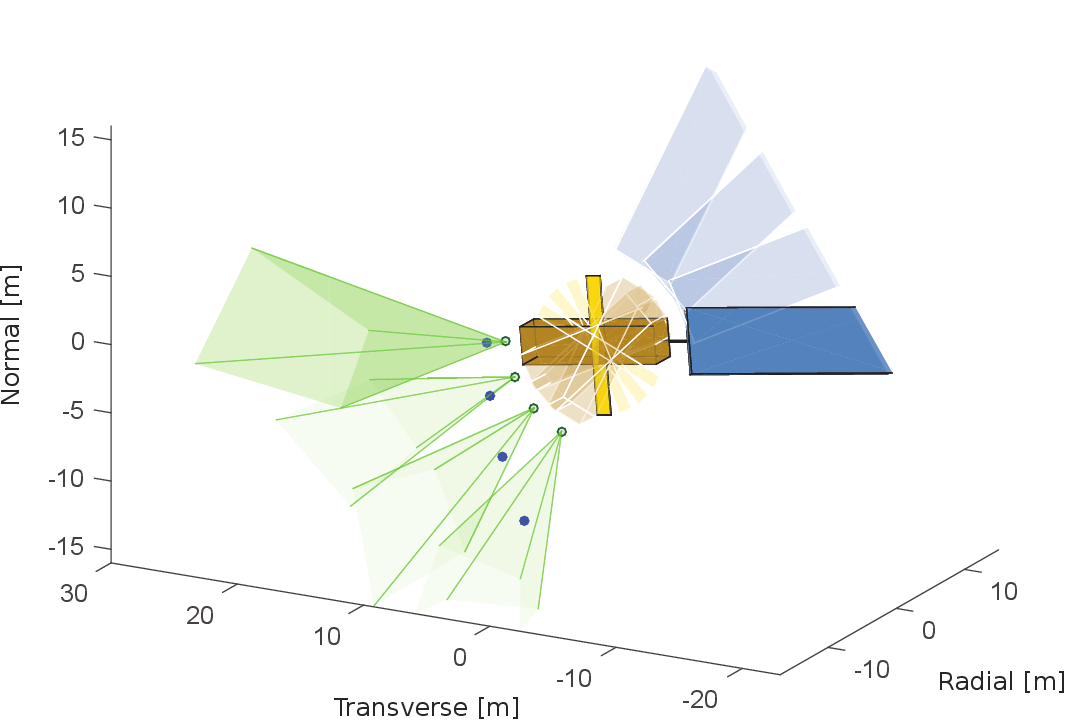}
	\caption{Illustration of the docking corridor rotation during the final approach to P2 (green point). The servicer center of mass (sCM) is marked by a blue point.}
	\label{RVDEnvi3D}
\end{figure}
\begin{table}[!t]
\centering\caption{GPOPS-II and Algorithm~\ref{alg1} performance for EnviSat RVD}\vspace{3mm}
\begin{tabular}{l c r}
\hline \hline
RVD to P1 & GPOPS-II & Algorithm~\ref{alg1}\\ \hline
CPU time & 5.767 s& 0.156 s $\quad$\\
Maneuver cost & 1.5697 & ${J}^*_{\boldhat{N}}\tau_s=$1.5774\\[1mm]
Normalized time-of-flight & 0.7852 & $\boldhat{N}\tau_s=$ 0.7977 \\
Normalized fuel consumption & 0.1961 & $\|{\mathbf{u}}^*_{\boldhat{N}}\|_1 \tau_s=$ 0.1949\\ \hline
RVD to P2 & GPOPS-II & Algorithm~\ref{alg1}\\ \hline
CPU time & 9.319 s& 0.205 s $\quad$\\
Maneuver cost & 2.8721 & ${J}^*_{\boldhat{N}}\tau_s=$ 2.9328\\[1mm]
Normalized time-of-flight & 0.7757 & $\boldhat{N}\tau_s=$ 0.7731 \\
Normalized fuel consumption & 0.5241 & $\|{\mathbf{u}}^*_{\boldhat{N}}\|_1 \tau_s=$ 0.5399\\
\hline\hline
\end{tabular}\label{gptab}
\end{table}
Similarly to what has been done in Section~\ref{solution}, the initial guess for the nonlinear solver is constructed from the unconstrained minimum-energy solution. The results of the comparison are summarized in Table~\ref{gptab}. It can be seen that the Algorithm~\ref{alg1} solution is about 40 times faster than that based on GPOPS-II, while the maneuver cost is approximately the same for the two methods, for RVD to either P1 or P2. In Figure~\ref{cGP}, the cost incurred by the two solutions is compared with the profile of ${J}^*_{N}\tau_s$ for case P2. It can be seen that the nonlinear solution is close to a local optimum of problem \eqref{fpt2}. This indicates that the constraint approximation scheme described in Section \ref{docking} is not overly conservative. On the whole, the obtained results clearly demonstrate the suitability of the proposed approach for autonomous RVD applications. In particular, in all our tests Algorithm~\ref{alg1} returned a solution in a fraction of a second, while the CPU time of the full enumeration, MILP and GPOPS-II approaches was always greater than 3 s. Considering that the sampling time of the guidance scheme is in the order of 10 s, and that spacecraft onboard computers are far less powerful than the employed hardware, the proposed method appears to be the right candidate for real-time implementation.

\begin{figure}[!t]
	\centering
	\includegraphics[width=0.65\textwidth]{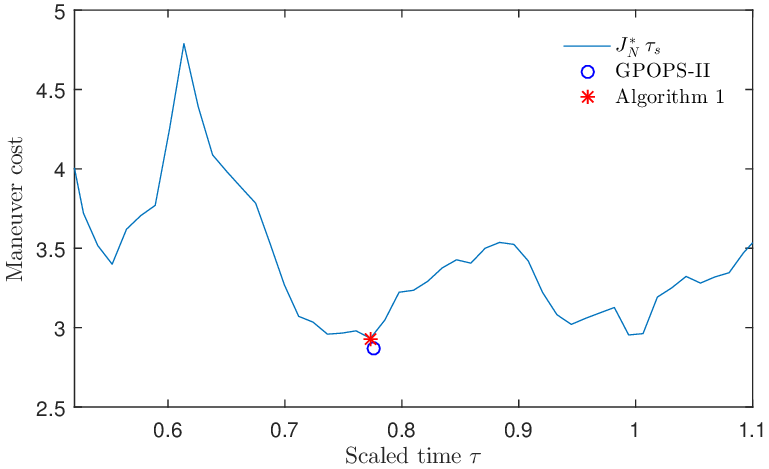}
	\caption{Profile of ${J}^*_{N}\tau_s$, together with the cost incurred by GPOPS-II and Algorithm~\ref{alg1}, for RVD to P2.}
	\label{cGP}
\end{figure}

\section{Conclusions}\label{conclusions}
A variable-horizon guidance scheme has been presented for spacecraft rendezvous and docking applications featuring a tumbling target. As opposed to approaches employing a fixed planning horizon, the proposed formulation provides the ability to identify favourable docking opportunities, which are singled out as local minima of a suitably chosen horizon-dependent cost function. A local optimization strategy has been developed for this new formulation, which is capable of finding high-performance solutions, while incurring a modest computational demand. The method also inherits other advantages of convex formulations, such as guaranteed convergence and ease of implementation. In view of these features, the proposed guidance scheme may be employed either as a standalone module for the autonomous planning (and re-planning) of optimized trajectories to be tracked by the control system, or as a baseline for the development of variable-horizon model predictive control strategies, tailored to uncooperative mission scenarios. The method has been demonstrated on a real-world scenario involving docking with the defunct EnviSat spacecraft, and found to provide remarkable improvements in terms of computational efficiency with respect to a nonlinear solver, while incurring only a negligible performance loss.

\bibliographystyle{aiaa-doi}
\bibliography{RVDtumbling_minimal}

\begin{thebibliography}{10}
\newcommand{\enquote}[1]{``#1''}

\bibitem{woffinden07}
Woffinden, D.~C. and Geller, D.~K., \enquote{Navigating the Road to Autonomous
  Orbital Rendezvous,} {\em Journal of Spacecraft and Rockets\/}, Vol.~44,
  No.~4, 2007, pp.~898--909.
\newblock \doi{10.2514/1.30734}.

\bibitem{kawano01}
Kawano, I., Mokuno, M., Kasai, T., and Suzuki, T., \enquote{Result of
  Autonomous Rendezvous Docking Experiment of {Engineering Test
  Satellite-VII},} {\em Journal of Spacecraft and Rockets\/}, Vol.~38, No.~1,
  2001, pp.~105--111.
\newblock \doi{10.2514/2.3661}.

\bibitem{rumford2003demonstration}
Rumford, T.~E., \enquote{Demonstration of Autonomous Rendezvous Technology
  {(DART)} Project Summary,} {\em Proceedings of the Society of Photo-Optical
  Instrumentation Engineers: Space Systems Technology and Operations\/}, Vol.
  5088, International Society for Optics and Photonics, Bellingham, WA, 2003,
  pp. 10--19.
\newblock \doi{10.1117/12.498811}.

\bibitem{mitchell2006gnc}
Mitchell, I., Gordon, T., Taskov, K., Drews, M., Luckey, D., Osborne, M., Page,
  L., Norris, H., and Shepperd, S., \enquote{{GNC} Development of the {XSS-11}
  Micro-satellite for Autonomous Rendezvous and Proximity Operations,} {\em
  29th AAS Guidance and Control Conference\/}, Breckenridge, CO, 2006.

\bibitem{weismuller2006gn}
Weismuller, T. and Leinz, M., \enquote{GN\&C Technology Demonstrated by the
  {Orbital Express} Autonomous Rendezvous and Capture Sensor System,} {\em 29th
  AAS Guidance and Control conference\/}, Breckenridge, CO, 2006.

\bibitem{boyarko2011}
Boyarko, G., Yakimenko, O., and Romano, M., \enquote{Optimal Rendezvous
  Trajectories of a Controlled Spacecraft and a Tumbling Object,} {\em Journal
  of Guidance, Control, and Dynamics\/}, Vol.~34, No.~4, 2011, pp.~1239--1252.
\newblock \doi{10.2514/1.47645}.

\bibitem{ventura2017}
Ventura, J., Ciarcià, M., Romano, M., and Walter, U., \enquote{Fast and
  Near-Optimal Guidance for Docking to Uncontrolled Spacecraft,} {\em Journal
  of Guidance, Control, and Dynamics\/}, Vol.~40, No.~12, 2017, pp.~3138--3154.
\newblock \doi{10.2514/1.G001843}.

\bibitem{ping13}
Lu, P. and Liu, X., \enquote{Autonomous Trajectory Planning for Rendezvous and
  Proximity Operations by Conic Optimization,} {\em Journal of Guidance,
  Control, and Dynamics\/}, Vol.~36, No.~2, 2013, pp.~375--389.
\newblock \doi{10.2514/1.58436}.

\bibitem{liu13}
Liu, X. and Lu, P., \enquote{Robust Trajectory Optimization for Highly
  Constrained Rendezvous and Proximity Operations,} {\em AIAA Guidance,
  Navigation, and Control (GNC) Conference\/}, 2013.
\newblock \doi{10.2514/6.2013-4720}.

\bibitem{mao2019successive}
Mao, Y., Szmuk, M., Xu, X., and Acikmese, B., \enquote{Successive
  Convexification: A Superlinearly Convergent Algorithm for Non-convex Optimal
  Control Problems,} 2019, arXiv:1804.06539v2 [math.OC].

\bibitem{bonalli19gusto}
Bonalli, R., Cauligi, A., Bylard, A., and Pavone, M., \enquote{GuSTO:
  Guaranteed Sequential Trajectory optimization via Sequential Convex
  Programming,} {\em 2019 International Conference on Robotics and Automation
  (ICRA)\/}, 2019, pp. 6741--6747.
\newblock \doi{10.1109/ICRA.2019.8794205}.

\bibitem{liu14}
Liu, X. and Lu, P., \enquote{Solving Nonconvex Optimal Control Problems by
  Convex Optimization,} {\em Journal of Guidance, Control, and Dynamics\/},
  Vol.~37, No.~3, 2014, pp.~750--765.
\newblock \doi{10.2514/1.62110}.

\bibitem{malyuta2021convex}
Malyuta, D., Reynolds, T.~P., Szmuk, M., Lew, T., Bonalli, R., Pavone, M., and
  Acikmese, B., \enquote{Convex Optimization for Trajectory Generation,} 2021,
  arXiv:2106.09125v1 [math.OC].

\bibitem{Foust20}
Foust, R., Chung, S.-J., and Hadaegh, F.~Y., \enquote{Optimal Guidance and
  Control with Nonlinear Dynamics Using Sequential Convex Programming,} {\em
  Journal of Guidance, Control, and Dynamics\/}, Vol.~43, No.~4, 2020,
  pp.~633--644.
\newblock \doi{10.2514/1.G004590}.

\bibitem{ping21}
Lu, P., \enquote{Convex–Concave Decomposition of Nonlinear Equality
  Constraints in Optimal Control,} {\em Journal of Guidance, Control, and
  Dynamics\/}, Vol.~44, No.~1, 2021, pp.~4--14.
\newblock \doi{10.2514/1.G005443}.

\bibitem{malyuta2020}
Malyuta, D., Reynolds, T., Szmuk, M., Acikmese, B., and Mesbahi, M.,
  \enquote{Fast Trajectory Optimization via Successive Convexification for
  Spacecraft Rendezvous with Integer Constraints,} {\em AIAA Scitech 2020
  Forum\/}, Orlando, FL, 2020.
\newblock \doi{10.2514/6.2020-0616}.

\bibitem{weiss15}
Weiss, A., Baldwin, M., Erwin, R.~S., and Kolmanovsky, I., \enquote{Model
  Predictive Control for Spacecraft Rendezvous and Docking: Strategies for
  Handling Constraints and Case Studies,} {\em IEEE Transactions on Control
  Systems Technology\/}, Vol.~23, No.~4, 2015, pp.~1638--1647.
\newblock \doi{10.1109/TCST.2014.2379639}.

\bibitem{zagaris18}
Zagaris, C., Park, H., Virgili-Llop, J., Zappulla, R., Romano, M., and
  Kolmanovsky, I., \enquote{Model Predictive Control of Spacecraft Relative
  Motion with Convexified Keep-Out-Zone Constraints,} {\em Journal of Guidance,
  Control, and Dynamics\/}, Vol.~41, No.~9, 2018, pp.~2054--2062.
\newblock \doi{10.2514/1.G003549}.

\bibitem{hartley13}
Hartley, E.~N., Gallieri, M., and Maciejowski, J.~M., \enquote{Terminal
  Spacecraft Rendezvous and Capture with {LASSO} Model Predictive Control,}
  {\em International Journal of Control\/}, Vol.~86, No.~11, 2013,
  pp.~2104--2113.
\newblock \doi{10.1080/00207179.2013.789608}.

\bibitem{leomanni20}
Leomanni, M., Bianchini, G., Garulli, A., Giannitrapani, A., and Quartullo, R.,
  \enquote{Orbit Control Techniques for Space Debris Removal Missions Using
  Electric Propulsion,} {\em Journal of Guidance, Control, and Dynamics\/},
  Vol.~43, No.~7, 2020, pp.~1259--1268.
\newblock \doi{10.2514/1.G004735}.

\bibitem{mammarella20}
Mammarella, M., Lorenzen, M., Capello, E., Park, H., Dabbene, F., Guglieri, G.,
  Romano, M., and Allgöwer, F., \enquote{An Offline-Sampling {SMPC} Framework
  With Application to Autonomous Space Maneuvers,} {\em IEEE Transactions on
  Control Systems Technology\/}, Vol.~28, No.~2, 2020, pp.~388--402.
\newblock \doi{10.1109/TCST.2018.2879938}.

\bibitem{breger08safe}
Breger, L. and How, J.~P., \enquote{Safe Trajectories for Autonomous Rendezvous
  of Spacecraft,} {\em Journal of Guidance, Control, and Dynamics\/}, Vol.~31,
  No.~5, 2008, pp.~1478--1489.
\newblock \doi{10.2514/1.29590}.

\bibitem{richards02}
Richards, A., Schouwenaars, T., How, J.~P., and Feron, E., \enquote{Spacecraft
  Trajectory Planning with Avoidance Constraints Using Mixed-Integer Linear
  Programming,} {\em Journal of Guidance, Control, and Dynamics\/}, Vol.~25,
  No.~4, 2002, pp.~755--764.
\newblock \doi{10.2514/2.4943}.

\bibitem{dicariano2012}
Di~Cairano, S., Park, H., and Kolmanovsky, I., \enquote{Model Predictive
  Control Approach for Guidance of Spacecraft Rendezvous and Proximity
  Maneuvering,} {\em International Journal of Robust and Nonlinear Control\/},
  Vol.~22, No.~12, 2012, pp.~1398--1427.
\newblock \doi{10.1002/rnc.2827}.

\bibitem{li17}
Li, Q., Yuan, J., Zhang, B., and Gao, C., \enquote{Model Predictive Control for
  Autonomous Rendezvous and Docking with a Tumbling Target,} {\em Aerospace
  Science and Technology\/}, Vol.~69, 2017, pp.~700--711.
\newblock \doi{10.1016/j.ast.2017.07.022}.

\bibitem{dong20}
Dong, K., Luo, J., Dang, Z., and Wei, L., \enquote{Tube-based Robust Output
  Feedback Model Predictive Control for Autonomous Rendezvous and Docking with
  a Tumbling Target,} {\em Advances in Space Research\/}, Vol.~65, No.~4, 2020,
  pp.~1158--1181.
\newblock \doi{10.1016/j.asr.2019.11.014}.

\bibitem{richards2006robust}
Richards, A. and How, J.~P., \enquote{Robust Variable Horizon Model Predictive
  Control for Vehicle Maneuvering,} {\em International Journal of Robust and
  Nonlinear Control\/}, Vol.~16, No.~7, 2006, pp.~333--351.
\newblock \doi{10.1002/rnc.1059}.

\bibitem{hartley2012}
Hartley, E.~N., Trodden, P.~A., Richards, A.~G., and Maciejowski, J.~M.,
  \enquote{Model Predictive Control System Design and Implementation for
  Spacecraft Rendezvous,} {\em Control Engineering Practice\/}, Vol.~20, No.~7,
  2012, pp.~695--713.
\newblock \doi{10.1016/j.conengprac.2012.03.009}.

\bibitem{louet99}
Louet, J. and Bruzzi, S., \enquote{{ENVISAT} Mission and System,} {\em IEEE
  1999 International Geoscience and Remote Sensing Symposium\/}, Vol.~3, 1999,
  pp. 1680--1682.
\newblock \doi{10.1109/IGARSS.1999.772059}.

\bibitem{fehse2003}
Fehse, W., {\em Automated Rendezvous and Docking of Spacecraft\/}, Cambridge
  Aerospace Series, Cambridge University Press, 2003.
\newblock \doi{10.1017/CBO9780511543388}, pp. 173-174.

\bibitem{clohessy60}
Clohessy, W.~H. and Wiltshire, R.~S., \enquote{Terminal Guidance System for
  Satellite Rendezvous,} {\em Journal of the Aerospace Sciences\/}, Vol.~27,
  No.~9, 1960, pp.~653--658.
\newblock \doi{10.2514/8.8704}.

\bibitem{gurobi}
{Gurobi Optimization LLC}, \enquote{Gurobi Optimizer Reference Manual,} 2021.

\bibitem{estable20}
Estable, S., Pruvost, C., Ferreira, E., Telaar, J., Fruhnert, M., Imhof, C.,
  Rybus, T., Peckover, G., Lucas, R., Ahmed, R., Oki, T., Wygachiewicz, M.,
  Kicman, P., Lukasik, A., Santos, N., Milhano, T., Arroz, P., Biesbroek, R.,
  and Wolahan, A., \enquote{Capturing and Deorbiting Envisat with an Airbus
  Spacetug. Results from the {ESA} {e.Deorbit} Consolidation Phase Study,} {\em
  Journal of Space Safety Engineering\/}, Vol.~7, No.~1, 2020, pp.~52--66.
\newblock \doi{10.1016/j.jsse.2020.01.003}.

\bibitem{Deloo15}
Deloo, J. and Mooij, E., \enquote{Active Debris Removal: Aspects of
  Trajectories, Communication and Illumination During Final Approach,} {\em
  Acta Astronautica\/}, Vol.~117, 2015, pp.~277--295.
\newblock \doi{10.1016/j.actaastro.2015.08.001}.

\bibitem{li2020real}
Li, H., Dong, Y., and Li, P., \enquote{Real-Time Optimal Approach and Capture
  of {ENVISAT} Based on Neural Networks,} {\em International Journal of
  Aerospace Engineering\/}, Vol.~2020, 2020.
\newblock \doi{10.1155/2020/8165147}.

\bibitem{kucharski2014attitude}
Kucharski, D., Kirchner, G., Koidl, F., Fan, C., Carman, R., Moore, C.,
  Dmytrotsa, A., Ploner, M., Bianco, G., Medvedskij, M., Makeyev, A., Appleby,
  G., Suzuki, M., Torre, J.-M., Zhongping, Z., Grunwaldt, L., and Feng, Q.,
  \enquote{Attitude and Spin Period of Space Debris Envisat Measured by
  Satellite Laser Ranging,} {\em IEEE Transactions on Geoscience and Remote
  Sensing\/}, Vol.~52, No.~12, 2014, pp.~7651--7657.
\newblock \doi{10.1109/TGRS.2014.2316138}.

\bibitem{patterson2014gpops}
Patterson, M.~A. and Rao, A.~V., \enquote{GPOPS-II: A MATLAB Software for
  Solving Multiple-Phase Optimal Control Problems Using Hp-Adaptive Gaussian
  Quadrature Collocation Methods and Sparse Nonlinear Programming,} {\em ACM
  Transactions on Mathematical Software (TOMS)\/}, Vol.~41, No.~1, 2014,
  pp.~1--37.
\newblock \doi{10.1145/2558904}.

\end{thebibliography}
\end{document}